\documentclass[journal]{IEEEtran}

\usepackage[latin1]{inputenc}
\usepackage{graphicx, amssymb, amsmath, amsfonts, amsthm, diagbox}
\usepackage[usenames,dvipsnames]{color}

\newtheorem{proposition}{Proposition}

\newcommand{\diag}{\mathop{\mathrm{diag}}}
\newcommand{\Mod}[1]{\ (\mathrm{mod}\ #1)}

\allowdisplaybreaks
\IEEEoverridecommandlockouts

\begin{document}
\title{Vector Perturbation Channel Inversion for \\ SWIPT MU-MISO Systems}

\author{Ioannis Krikidis, \IEEEmembership{Fellow, IEEE},  Constantinos Psomas,~\IEEEmembership{Senior Member,~IEEE}, and Symeon Chatzinotas,~\IEEEmembership{Senior Member,~IEEE}\vspace*{-5mm}
\thanks{I. Krikidis and C. Psomas are with the Department of Electrical and Computer Engineering, University of Cyprus, Cyprus (email: \{krikidis, psomas\}@ucy.ac.cy); S. Chatzinotas is with the
Interdisciplinary Centre for Security, Reliability and Trust (SnT), University of Luxembourg, L-1855, Luxembourg (email: symeon.chatzinotas@uni.lu). This work has received funding from the European Research Council (ERC) under the European Union's Horizon 2020 research and innovation programme (Grant agreement No. 819819) and in part by Luxembourg National Research Fund (FNR) under the CORE project RISOTTI under Grant C20/IS/14773976.}}

\maketitle

\begin{abstract}
This letter investigates the employment of vector-perturbation (VP) precoding to convey simultaneously information and energy in multiple-user multiple-input single-output (MU-MISO) downlink channel. We show that the conventional VP in addition to the information capacity benefits that provides to linear channel inversion techniques, it enhances the harvested energy at the receivers due to the extended symbol constellation. To further boost harvesting performance,  the proposed modified VP technique (named VP-EH) designs the VP integer offsets in order to maximize the delivered power. The proposed scheme incorporates an integer least square problem to find the closest lattice point to a point which is given by a Rayleigh quotient optimization problem. Finally, a convex combination between conventional VP and VP-EH is proposed to achieve a trade-off between maximizing information or energy. Theoretical and simulations results validate that VP is a promising technique to simultaneously convey information and energy in MU-MISO systems. 
\end{abstract}
\vspace{-0.2cm}

\vspace{-0.3cm}
\section{Introduction}

\IEEEPARstart{S}{imultaneous}  wireless information and power transfer (SWIPT) exploits the dual use of radio-frequency (RF) signals to ensure communication and energy sustainability in low power devices. Over the last few years, it has attracted a tremendous attention from the academia/industry and its potential benefits have been studied from different perspectives and communication scenarios \cite{BRU}. Specifically, a scenario of high practical interest is the integration of SWIPT in multiuser multiple-input single-output (MU-MISO) downlink systems, where a multiantenna access point (AP) aims to transfer simultaneously information/energy to single-antenna terminals.  Existing solutions consider linear precoding schemes that mainly employ channel inversion at the transmitter side (e.g., zero-forcing (ZF)) or non-linear information theoretic approaches that refer to dirty-paper coding (DPC) and suffer from a higher complexity \cite{RUI}. 

A fundamental weakness of the linear channel inversion techniques is that their associated MU-MISO capacity does not scale as the number of antennas and users increase simultaneously \cite{PEE}. To overcome this limitation, non-linear vector-perturbation (VP) techniques perturb the user data by integer offsets and resolve capacity scaling through modulo operation at the receivers \cite{HOC}. VP is based on lattice theory and its encoding process is equivalent to determine the closest lattice point to a given point; the associated searching refers to sphere encoding and thus the complexity decreases in comparison to DPC \cite{LEE, MAS}. Despite its efficiency, VP has been overlooked in the literature mainly due to the higher complexity in comparison to conventional channel inversion schemes. However, recent advances in electronics/hardware significantly reduce its implementation cost  and VP becomes a promising solution for the upcoming communication systems. On the other hand, the perturbation of user data by integer offsets increases the individual power of the transmitted symbols (i.e., the transmitted lattice point conveys more power than this one in the original constellation) which in combination to the VP power scaling factor enhance the total received power at the terminals; this property is promising for SWIPT applications.

Motivated by the above observations, this is the first paper that unlocks the potential benefits of VP in SWIPT MU-MISO systems.  We show that conventional VP increases the harvested energy at the receivers and outperforms ZF in terms of both information and energy transfer. In addition, a modified VP scheme is proposed (named VP-EH), which designs the integer offsets to further boost the energy harvesting (EH) performance. Specifically, the proposed scheme adjusts the VP offsets by determining the closest lattice point to the solution of a Rayleigh quotient optimization problem aiming to maximize EH. Theoretical results validate that VP-EH outperforms conventional VP in terms of average EH while exhibits asymptotically an exponential symbol error rate (SER) decay. Finally, a SWIPT-based VP scheme is introduced which elaborates a convex combination of the two extreme policies (VP and VP-EH) to control the trade-off between maximizing information transfer  and/or maximizing EH. 

\vspace{-0.2cm}
\section{System model}

We consider an MU-MISO donwnlink setup consisting of a single AP equipped with $M$ antennas and $K\leq M$ single-antenna users \cite{HOC}. The received signals at the $K$ users can be represented by an equivalent multiple-input multiple-output (MIMO) channel i.e.,   
\vspace{-0.2cm}   
\begin{align}
\pmb{y}=\pmb{H}\pmb{x}+\pmb{n},
\end{align}
where $\pmb{H} \in \mathbb{C}^{K\times M}$ is the channel matrix with elements $h_{k,m}\sim \mathcal{CN}(0,1)$ representing the flat frequency fading channel between the $m$-th transmit antenna and the $k$-th user \cite{LEE,MAS}; $\pmb{x}$ is the transmit signal vector with $\|\pmb{x}\|^2=P$, where $P$ is the transmitted power; $\pmb{n} \in \mathbb{C}^{K\times 1}\sim \mathcal{CN}(0,\sigma^2 \pmb{I}_K)$ is the additive white Gaussian noise (AWGN) vector. We assume a perfect channel state information at the AP (appropriate training and feedback optimization ensure this assumption \cite{KOB}). To generate the transmitted signal, the AP employs channel inversion precoding (ZF) with $\pmb{F}=\pmb{H}^H(\pmb{H}\pmb{H}^H)^{-1}$. In addition, according to the principles of VP, the AP perturbs the user data $\pmb{u}=[u_1,\cdots,u_K]^T$ by an integer offset vector $\tau \pmb{l}$, where $\tau \in \mathbb{R}^{+}$ and $\pmb{l} \in \mathbb{Z}^{K}+j\mathbb{Z}^{K}$; $u_k \in \mathcal{B}$, where $\mathcal{B}$ denotes the symbol constellation set. The transmitted signal can be written as   
\vspace{-0.1cm}   
\begin{align}
\pmb{x}=\sqrt{\frac{P}{\gamma}}\pmb{F}(\pmb{u}+\tau \pmb{l}), \label{system}
\end{align}
where $\gamma=\|\pmb{F}(\pmb{u}+\tau \pmb{l})\|^2$ is the transmit power scaling factor. By considering ZF precoding, the received signal vector is given by
\vspace{-0.2cm}   
\begin{align}
\pmb{y}=\pmb{H}\pmb{x}+\pmb{n}=\sqrt{\frac{P}{\gamma}}(\pmb{u}+\tau\pmb{l})+\pmb{n}.
\end{align}
Although practical SWIPT architectures split/orthogonalize the resources to convey information and energy (e.g., power splitting, time-switching) \cite{BRU}, here (for simplicity), we assume an ideal SWIPT architecture that allows to extract information/energy from the same signal without losses\footnote{The extension of the proposed scheme to practical SWIPT architectures is a straightforward extension.}. For the information transfer, the received signal is multiplied by the factor $\sqrt{\gamma/P}$ to eliminate the effect of the  transmit power scaling and then is driven to a modulo $\tau$ operator to remove the perturbation vector \cite{PEE}. The received signal at the $k$-th receiver is given by
\begin{align}
y_k=\sqrt{\frac{\gamma}{P}}y_k \Mod{\tau} =u_k+w_k,
\end{align}
where $w_k=\sqrt{\gamma/P}n_k$ is the equivalent AWGN component. On the other hand, by using a linear energy harvesting model, the total harvested power\footnote{We assume a normalized transmission time and therefore the notions of energy and power become equivalent and can used interchangeably.} is given by
\begin{align}
Q=\frac{P \|\pmb{u}+\tau \pmb{l}\|^2}{\gamma}\triangleq \frac{\|\pmb{u}'\|^2}{\gamma}. \label{vp_EH}
\end{align}
\vspace{-0.1cm}   
It is worth noting that the harvested power for the $k$-th user is given by $q_k=|u_k+\tau l_k|^2/\gamma$. From \eqref{vp_EH}, it can be seen that VP has a double effect on the energy harvesting; it affects the denominator through the transmit scaling factor as well as the nominator through the VP integer offset. 

\vspace{-0.4cm}
\subsection{Conventional VP precoding}

The conventional VP scheme designs the perturbation vector to minimize the transmit power scaling factor i.e.,  
\begin{align}
\pmb{l}_0&=\arg \min_{\pmb{l}'\in \mathbb{Z}^K+j\mathbb{Z}^K} \gamma 
\;\;= \arg \min_{\pmb{l}'\in \mathbb{Z}^K+j\mathbb{Z}^K} \|\pmb{F}(\pmb{u}+\tau \pmb{l}')\|^2 \nonumber \\
&=\arg \min_{\pmb{l}'\in \mathbb{Z}^K+j\mathbb{Z}^K} \|\pmb{F}\pmb{u}-(-\pmb{F}\tau \pmb{l}')\|^2. \label{VP1}
\end{align}
The above formulation  is an integer least square (ILS) problem and is equivalent to finding the closest lattice point to a given point. Although the ILS problem is NP-hard, sphere decoding\footnote{Since the SD is applied at the transmitter side, the process is well-known as {\it sphere encoder} in the literature \cite{HOC, MAS}.} (SD) is an efficient systematic search that can solve it in polynomial time. It is worth noting that several algorithms have been proposed in the literature to implement SD; in this work, we adopt the well-known Schnorr-Euchner (SE) scheme which is based on recursive tree searching without loss of generality \cite{MAS}. If $\pmb{z}=F_{\text{SE}}(L,\pmb{y},\pmb{A})$ denotes the SE algorithm that solves the standard ILS problem i.e., $\arg \min_{x}\|\pmb{y}-\pmb{A}\pmb{z}\|$  with dimension size $L$, the conventional VP scheme in \eqref{VP1} can be represented as \cite{MAS}
\begin{align}
\pmb{l}_0=F_{\text{SE}}(K,\pmb{F}\pmb{u},-\pmb{F}\tau).
\end{align} 

\vspace{-0.4cm}
\section{A SWIPT-based VP precoding scheme}\label{sec3}

The conventional VP scheme has been designed for information transfer and achieves a performance close to the Shannon capacity in MU-MISO downlink systems \cite{PEE,HOC}. In SWIPT systems, although it significantly affects/boosts EH at the receivers (see \eqref{vp_EH}), this aspect is not taken into account in the design. In this section,  we fill this gap and we investigate a VP scheme that controls the delivered power at the users.

In the first step of the SWIPT-based VP scheme, we relax the assumption that the AP transmits perturbated user data (belonging in $\mathcal{B}$) and investigate a normalized transmit vector $\pmb{u}_{\text{EH}}$ that maximizes the total received power $Q$; for this case, the transmitted signal becomes $\pmb{x}_{\text{EH}}=\sqrt{P}\pmb{F}\pmb{u}_{\text{EH}}/\|\pmb{F}\pmb{u}_{\text{EH}}\|$. More specifically, we introduce the following optimization problem 
\vspace{-0.2cm}   
\begin{align}
\pmb{u}_{\text{EH}}&=\arg \max_{\pmb{u}, \|\pmb{u}\|=1}\frac{\|\pmb{u}\|^2}{\pmb{u}^H \pmb{F}^H \pmb{F}\pmb{u}}=\arg \max_{\pmb{u}, \|\pmb{u}\|=1}\frac{\pmb{u}^H \pmb{u}}{\pmb{u}^H (\pmb{H}\pmb{H}^H)^{-1} \pmb{u}^H} \nonumber \\
&=\arg \min_{\pmb{u}, \|\pmb{u}\|=1}\frac{\pmb{u}^H (\pmb{H}\pmb{H}^H)^{-1} \pmb{u}^H}{\pmb{u}^H \pmb{u}}. \label{re_q}
\end{align} 
Since the matrix $\pmb{H}\pmb{H}^H$ is Hermitian and positive definite, the formulation in \eqref{re_q} is a (standard) Rayleigh quotient \cite{HOR}; if $\pmb{H}\pmb{H}^H=\pmb{V}\pmb{\Lambda}\pmb{V}^H$ is the eigenvalue decomposition of the matrix $\pmb{H}\pmb{H}^H$, where $\pmb{V}=[\pmb{v}_1,\ldots,\pmb{v}_K]\in \mathbb{C}^{K\times K}$, $\pmb{V}\pmb{V}^H=\pmb{V}^H\pmb{V}=\pmb{I}_K$ and $\pmb{\Lambda}=\diag(\lambda_1,\ldots,\lambda_K)$ with $\lambda_1\geq \lambda_2\geq\ldots \lambda_K\geq 0$, the optimal solution of problem \eqref{re_q} is $\pmb{u}_{\text{EH}}=\pmb{v}_1$ (i.e., the eigenvector corresponding to the largest eigenvalue\footnote{Since the matrix $\pmb{H}\pmb{H}^H$ is Hermitian and positive definite, its inverse admits also similar characteristics; the minimum eigenvalue of the matrix $(\pmb{H}\pmb{H}^H)^{-1}$ is the maximum eigenvalue of the matrix $\pmb{H}\pmb{H}^H$.}) which gives $Q=P\lambda_1$ . We note that any scaled version of the vector $\pmb{u}_{\text{EH}}$ i.e., $c\pmb{u}_{\text{EH}}$ with $c \in \mathbb{C}$ is a solution of \eqref{re_q} without affecting the maximum value of the total harvested power $Q$.  

Although the vector $c \pmb{u}_{\text{EH}}$ maximizes the total received power at the receivers, it does not take into account the fact that the user data $\pmb{u}$ are predefined and take values in the symbol constellation set $\mathcal{B}$. By using the degree of freedom of the offset perturbation vector $\tau \pmb{l}$, we perturb the user data in such a way that the perturbed data (lattice point) are as close as possible to $c \pmb{u}_{\text{EH}}$; it is worth noting that the scaling factor $c$ offers an extra degree of freedom to minimize the distance between these two points. This searching problem can be formulated as a modified ILS and is given by
\begin{align}
\pmb{l}_{\text{EH}}=\arg \min_{\pmb{l}'\in \mathbb{Z}^K+j\mathbb{Z}^K,c \in \mathbb{C}}\|\pmb{u}+\tau \pmb{l}'-c\pmb{u}_{\text{EH}}\|^2. \label{eh_l}
\end{align}
The problem in \eqref{eh_l} can be solved in two steps. In the first step, we compute $c$ that minimizes \eqref{eh_l} for a given perturbation vector $\pmb{l}'$ and then we solve a standard ILS problem by using the conventional SE algorithm. Specifically, by denoting $\xi(\pmb{l}')=\min_{c \in \mathbb{C}}\|\pmb{u}+\tau\pmb{l}'-c\pmb{u}_{\text{EH}}\|^2$, the problem in \eqref{eh_l} is written as
\vspace{-0.1cm}
\begin{align}
\pmb{l}_{\text{EH}}=\arg \min_{\pmb{l}'\in \mathbb{Z}^K+j\mathbb{Z}^K}\xi(\pmb{l}'). \label{eh_l2}
\end{align}
Since $\xi(\pmb{l})'$ is an unconstrained scalar optimization problem, the derivative of the objective function is equal to zero in the optimal point. Therefore, we have
\begin{align}
\frac{\partial \|\pmb{u}+\tau\pmb{l}'-c\pmb{u}_{\text{EH}}\|^2}{\partial c}=0 \Rightarrow c=\frac{(\pmb{u}+\tau \pmb{l}')^H \pmb{u}_{\text{EH}}}{\|\pmb{u}_{\text{EH}}\|^2}. \label{eh_2}
\end{align} 
Plugging \eqref{eh_2} into \eqref{eh_l2} and after some manipulations, we have
\begin{align}
\pmb{l}_{\text{EH}}&=\arg \min_{\pmb{l}'\in \mathbb{Z}^K+j\mathbb{Z}^K} \left \|\pmb{Z}\pmb{u}-(-\pmb{Z}\tau)\pmb{l}' \right\| \nonumber \\
&=F_{\text{SE}}\left(K, \pmb{Z}\pmb{u},-\pmb{Z}\tau \right),
\end{align}
where $\pmb{Z}=\pmb{I}_K-\frac{\pmb{u}_{\text{EH}}\pmb{u}_{\text{EH}}^H}{\|\pmb{u}_{\text{EH}}\|^2}=\pmb{I}_K-\pmb{u}_{\text{EH}}\pmb{u}_{\text{EH}}^H$ with $\|\pmb{u}_{\text{EH}}\|^2=1$.
The above modified VP scheme (called VP-EH) maximizes the delivered power but does not take into account the performance of information transfer in the design process; on the other hand, conventional VP has the opposite  objective. Inspired by these two extreme policies that focus either on the information or the energy transfer, we provide a new VP scheme (called VP-SWIPT) that achieves a trade-off between them. To control this trade-off, we introduce the convex combination between the two extreme lattice operation points i.e., $\pmb{\delta}(\eta)=\pmb{F}(\pmb{u}+\tau\pmb{l}_0) \eta+\pmb{F}(\pmb{u}+\tau\pmb{l}_{\text{EH}})(1-\eta)$, where $\eta \in [0,1]$ is a design parameter that represents a desired operation point at the linear segment between the two extreme lattice points. Then, for a specific value of $\eta$, we search for a lattice point that is the closest to this operation point. Specifically, we consider the following ILS problem which is solved through appropriate parametrization of the SE algorithm  
\begin{align}
\pmb{l}_{\text{SW}}&=\arg\min_{\pmb{l}'\in \mathbb{Z}^K+j\mathbb{Z}^K} \|\pmb{F}(\pmb{u}+\tau \pmb{l}')-\pmb{\delta}(\eta)\|^2 \nonumber \\
&=\arg\min_{\pmb{l}'\in \mathbb{Z}^K+j\mathbb{Z}^K} \|(\pmb{F}\pmb{u}-\pmb{\delta}(\eta))-(-\pmb{F}\tau)\pmb{l}'\|^2 \nonumber \\
&=F_{\text{SE}}\left(K, \pmb{F}\pmb{u}-\pmb{\delta}(\eta),-\pmb{F}\tau \right).
\end{align}
It is worth noting that the values $\eta=0$ and $\eta=1$ correspond to the VP and the VP-EH scheme, respectively.

\vspace{-0.3cm}
\section{Performance analysis}

In this section, we study theoretically the performance of the proposed VP-based schemes. 

\vspace{-0.3cm}
\subsection{Conventional VP precoding}

We study the conventional VP in terms of average harvested power across both users and channel realizations. Due to the complexity associated with the integer values of $\pmb{l}$, we investigate a tractable approximation/bound. Specifically, we have 
\begin{align}
\mathcal{E}_{\text{VP}}=\mathbb{E}(q_k)=\mathbb{E} \left(\frac{|\tilde{u}_k|^2}{\gamma} \right)\approx \frac{\mathbb{E}(|\tilde{u}_k|^2)}{\mathbb{E}(\gamma)}. \label{ff1}
\end{align}
Due to the symmetry between the real and the imaginary parts, the nominator can be written as
\begin{align}
\mathbb{E}(|\tilde{u}_k|^2)&=\mathbb{E}(|u_k+\tau l_{0,k}|^2)= \sum_{u' \in \mathcal{B}} 
\sum_{l' \in \mathbb{Z}+j\mathbb{Z}^K} \!\!\mathbb{P}(u',l')|u'+\tau l'|^2 \nonumber \\
&\approx 2 \sum_{u_r' \in \Re(\mathcal{B})} 
\sum_{l_r'=-1}^{l_r'=1} \mathbb{P}(u_r',l_r')(u_r'+\tau l_r')^2, \label{joint}
\end{align}
where the notation $\mathcal{R}(\cdot)$ denotes the real part of a variable or set, $u_r'=\mathcal{R}(u')$, $l_r'=\mathcal{R}(l')$, $\mathbb{P}(u_r',l_r')$ is the joint distribution of the variables $u_r'$ (transmitted symbol - real part) and $l_r'$ (associated integer perturbation offset - real part). We note that the perturbation offset takes values in a small set i.e., $\{-1,0,1\}$ and the joint distribution has an anti-symmetry property; this observation has been discussed in \cite{YUE} (through numerical studies) and validated here as well (Section \ref{numerical}, Table \ref{table1}). On the other hand, by using the seminal result in \cite{RYA}, the denominator is approximated by the following lower bound 
\begin{align}
\mathbb{E}(\gamma)&\geq \frac{K \Gamma(K+1)^{1/K}}{(K+1)\pi} \tau^2 \det(\pmb{F}^H \pmb{F})^{1/K} \nonumber \\
&=\frac{K \Gamma(K+1)^{1/K}}{(K+1)\pi} \tau^2 \det(\pmb{H}\pmb{H}^H)^{-1/K}  \nonumber \\
&=\frac{K \Gamma(K+1)^{1/K}}{(K+1)\pi} \tau^2 \prod_{m=0}^{K-1}\frac{\Gamma(M-\frac{1}{K}-m)}{\Gamma(M-m)}, \label{res4}
\end{align}
where $\Gamma(\cdot)$ denotes the gamma function. In the following proposition, we compare (qualitative comparison) the VP scheme with the conventional ZF scheme in terms of SWIPT; it is worth noting that the information transfer performance of the VP scheme is well known in the literature and therefore is not further discussed \cite{HOC}.  
\begin{proposition}\label{prop1}
The VP precoding scheme outperforms conventional ZF in terms of both information and power transfer.
\end{proposition}
\begin{proof}
See Appendix \ref{app1}.
\end{proof}

\vspace{-0.5cm}
\subsection{VP-EH precoding}

The EH performance of the VP-EH scheme approximates the performance of the ideal scenario where the AP transmits energy signals i.e., $\pmb{u}_{\text{EH}}$ (without the constraint of $\mathcal{B}$). By following the discussion in Section \ref{sec3}, the harvested power at the $k$-th user is written as 
\begin{align}
\mathcal{E}_{\text{EH}}\approx \frac{P \mathbb{E}(\lambda_1)}{K}=\frac{P}{K}\int_{0}^{\infty}(1-F_{\lambda}(x))dx, \label{res5}
\end{align}
where $F_{\lambda} (x)$ is the cumulative distribution function (CDF) of the maximum eigenvalue of a complex Wishart matrix, written as \cite[Sec. III.A]{CHI}
\begin{align}
F_{\lambda}(x)=\phi \det(\pmb{A}(x)),
\end{align}
where the elements of the matrix $\pmb{A}(x)$ and the constant $\phi$ are given respectively by
\begin{align}
A_{i,j}(x)&=\int_{0}^{x} t^{M+K-i-j}\exp(-t)dt \nonumber \\
&=\gamma(M+K-i-j+1,x),
\end{align}
and
\vspace{-0.3cm}
\begin{align}
&\phi=1/\prod_{i=1}^{K}(M-i)!(K-i)!,
\end{align}
where $\gamma(s,x)=\int_{0}^{x}t^{s-1}\exp(-t)dt$ in the lower incomplete gamma function.

Then, we study the information transfer performance of the VP-EH scheme. By using the approximation $(\pmb{u}+\tau \pmb{l}_{\text{EH}})\approx c \pmb{u}_{\text{EH}}$, we have $\gamma=\|\pmb{F}(\pmb{u}+\tau \pmb{l}_{\text{EH}})\|^2\approx \|c\pmb{F} \pmb{u}_{\text{EH}}\|^2= \|c\pmb{F} \pmb{v}_1\|^2=|c|^2 \pmb{v}_1^H \pmb{V}\pmb{\Lambda}^{-1}\pmb{V}^H \pmb{v}_1=|c|^2 \lambda_K$; therefore, the information transfer branch can be simplified to the following equivalent channel model i.e., 
\begin{align}
y_k=u_k+|c|\sqrt{\frac{\lambda_K}{P}}n_k, \label{equiv}
\end{align}
where $\lambda_K$ is the minimum eigenvalue of the complex Wishart matrix $\pmb{H}\pmb{H}^H$. The equivalent channel in \eqref{equiv} can be used to study the asymptotic behaviour ($P\rightarrow \infty$) of the VP-EH scheme as it is stated in following proposition. 

\vspace{-0.2cm}
\begin{proposition}\label{prop2}
The outage probability of the EH-based VP scheme is proportional to $\exp(-K P)$ (exponential decay). Since the conventional VP exhibits a diversity gain equal to $M$, the VP-EH scheme outperforms conventional VP asymptomatically in both information and energy transfer.  
\end{proposition}

\begin{proof}
See Appendix \ref{app2}
\end{proof}

It is worth noting that the above behaviour is observed in the high signal-to-noise ratio (SNR) regime; the outage probability decays exponentially with a rate of decay $1-\exp(-K)$. For the low and intermediate SNRs (which is the regime of practical interest), the conventional VP outperforms VP-EH in terms of outage probability.  

\vspace{-0.4cm}
\subsection{VP-SWIPT precoding}

The VP-SWIPT scheme is a convex combination of the VP and VP-EH schemes and therefore provides a performance between them. The performance depends on the parameter $\eta$ (in a non-linear way) as well as the density of the lattice (which provides the closest transmitted vector $\pmb{u}+\tau \pmb{l}_{\text{SW}}$); numerical results in the next section demonstrate the non-linear relationship of the VP-SWIPT  with the VP/VP-EH schemes.

\begin{figure}
\centering
\includegraphics[width=0.6\linewidth]{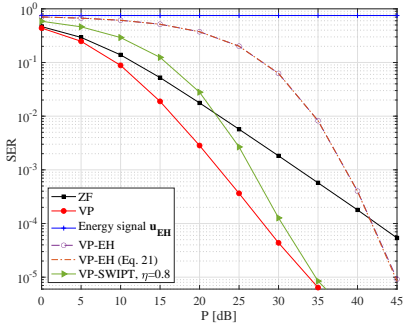}
\vspace{-0.3cm}
\caption{SER performance versus transmit power $P$; $M=K=2$, $\sigma^2=1$, $\eta=0.8$, $4$-QAM, $\tau=4$.}\label{fig1}
\end{figure}

\begin{figure}
\centering
\includegraphics[width=0.65\linewidth]{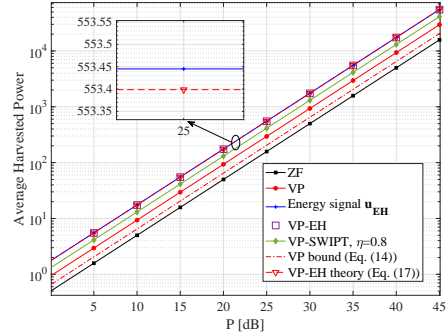}
\vspace{-0.3cm}
\caption{Avreage harvested power versus transmit power $P$; $M=K=2$, $\eta=0.8$, $4$-QAM, $\tau=4$.}\label{fig2}
\end{figure}

\begin{figure}
	\centering
	\includegraphics[width=0.7\linewidth]{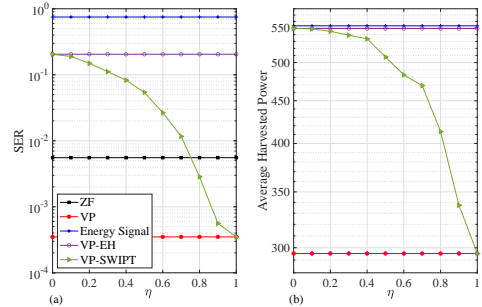}
	\vspace{-0.3cm}
	\caption{(a) SER versus $\eta$, (b) Average harvested power versus $\eta$; $P=25$ dB; $M=K=2$, $\sigma^2=1$, $4$-QAM, $\tau=4$.}\label{fig3}
\end{figure}

\vspace{-0.4cm}
\section{Numerical results \& discussion}\label{numerical}

Our setting assumes $M=K=2$, $4$-QAM with $\tau=4$, $\sigma^2=1$, and $\eta=0.8$. Table \ref{table1} shows the joint probability distribution of perturbation for the considered setup and validates our discussion in \eqref{joint}.

Fig. \ref{fig1} depicts the average SER performance versus the transmit power for the VP-based schemes considered; the ZF scheme and the transmission of energy signals $\pmb{u}_{\text{EH}}$ are used as performance benchmarks. The VP scheme provides the best SER performance in the SNR regime of interest and ensures full diversity in comparison to ZF scheme. On the other hand, the VP-EH scheme follows an exponential decay and outperforms ZF at high SNRs; the slopes of the curves validate the asymptotic gain of the VP-EH scheme against the conventional VP scheme (Proposition \ref{prop2}). It can be also shown that the the VP-SWIPT scheme provides a performance between VP and VP-EH schemes and approximates VP at high SNRs.  Fig. \ref{fig2} compares the VP schemes in terms of average harvested power. It can be observed that the VP-EH scheme approximates the performance of the EH benchmark (transmission of $\pmb{u}_{\text{EH}}$) while providing a gain of $2.5$ dB and $7$ dB in comparison to VP and ZF schemes, respectively. The VP-SWIPT provides a performance  between the two extreme policies, which is inline with the information transfer observations in Fig. \ref{fig1}. Our curves validate also our theoretical derivations in \eqref{ff1} and \eqref{res5}. 

Fig. \ref{fig3} focuses on the VP-SWIPT scheme  and studies the impact of the parameter $\eta$. It can be seen that VP-SWIPT scheme provides an information/energy performance between the two extreme policies (VP and VP-EH) by following a non-linear dependency with $\eta$. Finally, Fig. \ref{fig4} deals with the EH performance for a setup with $M=K=4$ and $16$-QAM; the observations follow the discussion of Fig. \ref{fig2} but we can observe a higher gain for the VP-EH mainly due to the higher modulation (i.e., $4.5$ dB against VP and $10$ dB against ZF).  Our theoretical and simulation results validate that VP introduces a new degree of freedom to control information and energy transfer and significantly outperforms conventional linear counterparts; it is a powerful tool for SWIPT MU-MISO systems. Other channel inversion schemes and/or channel models can be considered for future work.

\begin{figure}
	\centering
	\includegraphics[width=0.6\linewidth]{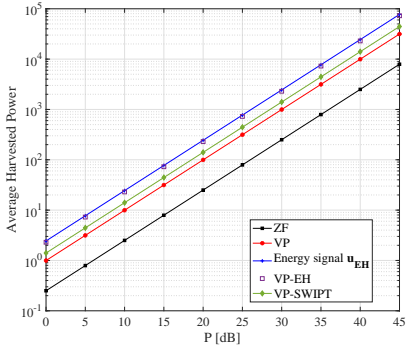}
	\vspace{-0.3cm}
	\caption{Average harvested power versus $P$; $M=K=4$, $16$-QAM, $\tau=8$.} \label{fig4}\vspace{-0.2cm}
\end{figure}

\begin{table}[t]\centering\caption{Joint probability distribution of perturbation $\mathbb{P}(l_{0,k} = x, u_{k} = y)$; $M=K=2$, $4$-QAM.}\label{table1}
	\vspace{-0.3cm}
\setlength\tabcolsep{4.5pt}\renewcommand\arraystretch{1.5}
\begin{tabular}{|c|c|c|c|}\hline
\diagbox[width=\dimexpr\textwidth/14\relax,height=0.85cm]{$u_k$}{$l_{0,k}$}
& $-1$ & $0$ & $+1$\\\hline
$-1$ & $0.0098$ & $0.3972$ & $0.0915$\\\hline
$+1$ & $0.0915$ & $0.3972$ & $0.0098$\\\hline
\end{tabular}\vspace{-0.3cm}
\end{table}
\vspace{-0.4cm}
\appendix
\subsection{Proof of Proposition 1}\label{app1}
We note that VP is equivalent to ZF for the case where $\pmb{l}=\pmb{0}_K$ (with $\pmb{0}_K$ denoting the zero vector). The VP scheme is based on the minimization of the power scaling factor $\gamma$ and therefore outperforms ZF in terms of information transfer (i.e., sum-capacity,  diversity etc); this has been proven in several previous works \cite{PEE,HOC}. As for the power transfer, the VP scheme achieves a harvested power $Q$ which is larger or equal to the ZF scheme. Specifically, the denominator in \eqref{vp_EH} follows the discussion of the information transfer i.e., $\gamma_{\text{VP}}\leq \gamma_{\text{ZF}}$; for the nominator, we focus on the case with a non-zero perturbation vector ($l_k \neq 0$). By assuming a $N$-QAM modulation with maximum and minimum absolute values $|c_{\text{max}}|$ and $|c_{\text{min}}|$ respectively, a minimum distance in the constellation $\Delta$ and $\tau=2|c_{\text{max}}|+\Delta$ (similar to \cite[eq. 9]{HOC}), we have
\vspace{-0.2cm}
\begin{align}
&|u_k+\tau l_k| \geq ||c_{\text{min}}|-\tau| = |c_{\text{max}}|+\frac{\sqrt{N}}{2} \Delta >|c_{\text{max}}|  \\
&\Rightarrow \|\pmb{u}+\tau \pmb{l}\|^2 > \|\pmb{u}\|^2 \Rightarrow Q_{\text{VP}}=\frac{\|\pmb{u}+\tau \pmb{l}\|^2}{\gamma_{\text{VP}}}
>Q_{\text{ZF}}=\frac{\|\pmb{u}\|^2}{\gamma_{\text{ZF}}}.
\end{align}
By considering all the possible values of $l_k \in \mathbb{Z}+j\mathbb{Z}$, we prove that $\mathbb{E}(Q_{\text{VP}})\geq \mathbb{E}(Q_{\text{ZF}})$.

\vspace{-0.2cm}
\subsection{Proof of Proposition 2}\label{app2}
We study the asymptotic performance of the EH-based VP precoding scheme by using the outage probability metric. Specifically, by using the equivalent channel in \eqref{equiv}, the outage probability can be written as
\vspace{-0.2cm}
\begin{align}
&P_{\text{out}}=\mathbb{P}(\log_2(1+\text{SNR})<r)=\mathbb{P}\left(\frac{P}{|c|^2 \lambda_K}<2^r-1 \right) \nonumber \\
&=\mathbb{P}\left(\! \lambda_K\!>\!\frac{P}{|c|^2 (2^r-1)} \!\right)\!=\!\phi \det\left(\pmb{B}\left(\frac{P}{|c|^2 (2^r-1)}\right)\right),  \label{re1}
\end{align}
where $r$ is the requested spectral efficiency, the expression in \eqref{re1} is given in \cite[Sec. III.A]{CHI} and $B_{i,j}(x)=\Gamma(M+K-i-j+1,x)$ with $\Gamma(s,x)=\int_{x}^{\infty}t^{s-1}\exp(-t)dt$ denoting the upper incomplete gamma function. Since $\Gamma(s,x)=(s-1)!\exp(-x)\sum_{m=0}^{s-1}\frac{x^m}{m!}$ \cite[8.354.4]{GRA} for $s \in \mathbb{Z}$, we can write
\begin{align}
P_{\text{out}} &= \phi\sum_{\kappa \in p} \text{sgn}(\kappa) \prod_{i=1}^K (M\!+\!K\!-\!i\!-\!\kappa_i)! \exp\left(\!\!-\frac{P}{|c|^2 (2^r-1)}\!\right)\nonumber\\
&\qquad\times \sum_{m=0}^{M+K-i-\kappa_i} \frac{1}{m!} \left(\frac{P}{|c|^2 (2^r-1)}\right)^m\nonumber\\
&\doteq \phi\exp(-KP) \sum_{\kappa \in p} \text{sgn}(\kappa) \prod_{i=1}^K (M+K-i-\kappa_i)! \nonumber\\
&\qquad\times \sum_{m=0}^{M+K-i-\kappa_i} \frac{P^m}{m!}\;\;\;\;\;\;\;\;\;\propto \exp(-K P),
\end{align}
where $p = \{\kappa = (\kappa_1,\dots,\kappa_K)\}$ is the set of $K!$ permutations of $(1,\dots,K)$, $\text{sgn}(\kappa)$ is the signature of the permutation $\kappa$, and the notation $\doteq$ denotes asymptotic expression when $P\rightarrow \infty$. It is worth noting that $|c|^2$ is a random variable; however, we can show that is bounded and therefore is considered as a constant asymptotically i.e., $|c|=|(\pmb{u}+\tau \pmb{l}_{\text{EH}})^H \pmb{u}_{\text{EH}}|\leq \|u+\tau \pmb{l}_{\text{EH}}\| \|\pmb{u}_{\text{EH}}\|=\|\pmb{u}+\tau \pmb{l}_{\text{EH}}\|\leq \|\pmb{u}\|+\tau \|\pmb{l}_{\text{EH}}\| < \infty$ since the elements in $\pmb{u}$, $\pmb{l}_{\text{EH}}$ are finite.

\end{document}